\documentclass[letterpaper,11pt]{article}
\usepackage{url}
\usepackage{amsthm}
\usepackage{fullpage}
\usepackage{graphicx}

\usepackage{amsmath,amssymb,xspace,cite}
\usepackage[usenames]{color}
\usepackage{hyperref}

\usepackage{algpseudocode}

\DeclareMathOperator*{\E}{\mathbb{E}}
\let\Pr\relax
\DeclareMathOperator*{\Pr}{\mathbb{P}}

\DeclareMathOperator{\poly}{poly}

\newcommand{\eqdef}{\mathbin{\stackrel{\rm def}{=}}}
\newcommand{\eps}{\varepsilon}

\newcommand{\R}{\mathbb{R}}
\newcommand{\Z}{\mathbb{Z}}

\newtheorem{theorem}{Theorem}
\newtheorem{corollary}[theorem]{Corollary}

\newtheorem{remark}[theorem]{Remark}
\newtheorem{question}[theorem]{Question}

\newcommand{\proofbelow}{3pt}
\newcommand{\afterproof}{\hfill $\blacksquare$ \par \vspace{\proofbelow}}
\renewenvironment{proof}{\noindent\textbf{Proof.}\,}{\afterproof}
\newenvironment{proofof}[1]{\noindent\textbf{Proof} \,(of #1).\,}{\afterproof}

\newcommand{\QuestionName}[1]{\label{que:#1}}

\renewcommand{\lg}{\log}
\newcommand{\Perm}{\pi}

\newcommand{\Question}[1]{Question~\ref{que:#1}}

\begin{document}

\author{Kasper Green Larsen\thanks{Aarhus University. \texttt{larsen@cs.au.dk}. Supported by Center for Massive Data Algorithmics, a Center of the Danish National Research Foundation, grant DNRF84.}
  \and Jelani Nelson\thanks{Harvard University. \texttt{minilek@seas.harvard.edu}. Supported by NSF
  CAREER award CCF-1350670.}
\and Huy L. Nguy$\tilde{\hat{\mbox{e}}}$n\thanks{Princeton
  University. \texttt{hlnguyen@cs.princeton.edu}. Supported in part by
  NSF CCF-0832797, an IBM Ph.D. fellowship and a Siebel Scholarship.}}

\title{Time lower bounds for nonadaptive turnstile streaming algorithms}

\maketitle

\begin{abstract}
We say a turnstile streaming algorithm is {\em non-adaptive} if, during updates, the memory cells written and read depend only on the index being updated and random coins tossed at the beginning of the stream (and not on the memory contents of the algorithm). Memory cells read during {\em queries} may be decided upon adaptively. All known turnstile streaming algorithms in the literature are non-adaptive.

We prove the first non-trivial update time lower bounds for both randomized and deterministic turnstile streaming algorithms, which hold when the algorithms are non-adaptive. While there has been abundant success in proving space lower bounds, there have been no non-trivial update time lower bounds in the turnstile model. Our lower bounds hold against classically studied problems such as heavy hitters, point query, entropy estimation, and moment estimation. In some cases of deterministic algorithms, our lower bounds nearly match known upper bounds.
\end{abstract}

\section{Introduction}
\label{sec:intro}
In the turnstile streaming model of computation \cite{Muthukrishnan05} there is some vector $v\in\R^n$ initialized to $\vec{0}$, and we must provide a data structure that processes coordinate-wise updates to $v$. An update of the form $(i,\Delta)$ causes the change $v_i\leftarrow v_i + \Delta$, where $\Delta\in\{-M,\ldots,M\}$. Occasionally our data structure must answer queries for some function of $v$. In many applications $n$ is extremely large, and thus it is desirable to provide a data structure with space consumption much less than $n$, e.g.\ polylogarithmic in $n$. For example, $n$ may be the number of valid IP addresses ($n = 2^{128}$ in IPv6), and $v_i$ may be the number of packets sent from source IP address $i$ on a particular link. A query may then ask for the support size of $v$ (the ``distinct elements'' problem \cite{FlajoletM85}), which was used for example to estimate the spread of the Code Red worm after filtering the packet stream based on the worm's signature \cite{EstanVF06,Moore01}. Another query may be the $\ell_2$ norm of $v$ \cite{AMS99}, which was used by AT{\&}T as part of a packet monitoring system \cite{KrishnamurthySZC03, ThorupZ12}. In some examples there is more than one possible query to be asked; in ``point query'' problems a query is some $i\in[n]$ and the data structure must output $v_i$ up to some additive error, e.g.\ $\eps\|v\|_p$ \cite{CCF04,CM05}. Such point query data structures are used as subroutines in the heavy hitters problem, where informally the goal is to output all $i$ such that $v_i$ is ``large''. If the data structure is linearly composable (meaning that data structures for $v$ and $v'$ can be combined to form a data structure for $v-v'$), heavy hitters data structures can be used for example to detect trending topics in search engine query streams \cite{GoogleTrends, GoogleZeitgeist, CCF04}. In fact the point query structure \cite{CCF04} has been implemented in the log analysis language Sawzall at Google \cite{PikeDGQ05}.

Coupled with the great success in providing small-space data structures for various turnstile streaming problems has been a great amount of progress in proving space lower bounds, i.e.\ theorems which state that {\em any} data structure for some particular turnstile streaming problem must use space (in bits) above some lower bound. For example, tight or nearly tight space lower bounds are known for the distinct elements problem \cite{AMS99,Woodruff04,WoodruffThesis,JayramKS08}, $\ell_p$ norm estimation \cite{AMS99,BJKS04,CKS03,Woodruff04,KNW10a,Gronemeier09,Jayram09,JW13,WoodruffThesis,JayramKS08}, heavy hitters \cite{JowhariST11}, entropy estimation \cite{ChakrabartiCM10,KNW10a}, and several other problems.

While there has been much previous work on understanding the space required for solving various streaming problems, much less progress has been made regarding update time complexity: the time it takes to process an update in the stream.  This is despite strong motivation, since in several applications the data stream may be updated at an extremely fast rate; for example, \cite{ThorupZ12} reported in their application for $\ell_2$-norm estimation that their system was constrained to spend only 130 nanoseconds per packet to keep up with high network speeds. Of course, without any space constraint fast update time is trivial: store $v$ in memory explicitly and spend constant time to add $\Delta$ to $v_i$ after each update. Thus an interesting data structural issues arises: how can we simultaneously achieve small space and low update time for turnstile streaming data structures?  Surprisingly, very little is understood about this question for any turnstile streaming problem.  For some problems we have very fast algorithms (e.g.\ constant update time for distinct elements \cite{KNW10b}, and also for $\ell_2$ estimation \cite{ThorupZ12}), whereas for others we do not (e.g. super-constant time for $\ell_p$ estimation for $p\neq 2$ \cite{KNPW11} and heavy hitters problems \cite{CCF04,CM05}), and we do not have proofs precluding the possibility that a fast algorithm exists. Indeed, the only previous update time lower bounds for streaming problems are those of Clifford and Jalsenius \cite{CJ11} and Clifford, Jalsenius, and Sach \cite{CJS13} for streaming multiplication and streaming Hamming distance computation, respectively. These problems are significantly harder than e.g.\ $\ell_p$ estimation (the lower bounds proved apply even for super-linear space usage) and there appears to be no way of extending their approach to obtain lower bounds for the problems above. Also, the time lower bounds in \cite{CJ11,CJS13} are for the sum of the update and query time, whereas in many applications such as network traffic monitoring, updates are much more frequent and thus it is more acceptable for queries to be slower. Lastly but importantly, the problems considered in \cite{CJ11,CJS13} were not in the turnstile model.

A natural model for upper bounds in the streaming literature (and also for data structures in general), is the word RAM model: basic arithmetic operations on machine words of $w$ bits each cost one unit of time. In the data structure literature, one of the strongest lower bounds that can be proven is in the cell probe model \cite{MP69,Fredman78,Yao81}, where one assumes that the data structure is broken up into $S$ words each of size $w$ bits and cost is only incurred when the user reads a memory cell from the data structure ($S$ is the space). Lower bounds proven in this model hold against any algorithm operating in the word RAM model. In recent years there has been much success in proving data structure lower bounds in the cell probe model (see for example \cite{PatrascuThesis, LarsenThesis}) using techniques from communication complexity and information theory. Can these techniques be imported to obtain time lower bounds for clasically studied turnstile streaming problems?

\begin{question}\QuestionName{time-lb}
Can we use techniques from cell probe lower bound proofs to lower bound the update time complexity for classical streaming problems?
\end{question}

Indeed the lower bounds for streaming multiplication and Hamming distance computation were proved using the \emph{information transfer} technique of P{\v a}tra{\c s}cu and Demaine \cite{patrascuInformationTransfer}, but as mentioned, this approach appears useless against the turnstile streaming problems above.

There are a couple of obvious problems to attack for \Question{time-lb}. For example, for many streaming problems one can move from, say, $2/3$ success probability to $1-\delta$ success probability by running $\Theta(\log(1/\delta))$ instantiations of the algorithm in parallel then outputting the median answer. This is possible whenever the output of the algorithm is numerical, such as for example distinct elements, moment estimation, entropy estimation, point query, or several other problems. Note that doing so increases both the space and update time complexity of the resulting streaming algorithm by a $\Theta(\log(1/\delta))$ factor. Is this necessary? It was shown that the blowup in {\em space} is necessary for several problems by Jayram and Woodruff \cite{JW13}, but absolutely nothing is known about whether the blowup in time is required. Thus, any non-trivial lower bound in terms of $\delta$ would be novel.

Another problem to try for \Question{time-lb} is the heavy hitters problem. Consider the $\ell_1$ heavy hitters problem: a vector $v$ receives turnstile updates, and during a query we must report all $i$ such that $|v_i| \ge \eps \|v\|_1$. Our list is allowed to contain some false positives, but not ``too false'', in the sense that any $i$ output should at least satisfy $|v_i| \ge (\eps/2)\|v\|_1$. Algorithms known for this problem using non-trivially small space, both randomized \cite{CM05} and deterministic \cite{GM07,NNW14}, require $\tilde{\Omega}(\log n)$ update time. Can we prove a matching lower bound?

\paragraph{Our Results.}
In this paper we present the first update time lower bounds for a range of classical turnstile streaming problems. The lower bounds apply to a restricted class of randomized streaming algorithms that we refer to as randomized \emph{non-adaptive} algorithms.
We say that a randomized streaming algorithm is {\em non-adaptive} if:
\begin{itemize}
\item Before processing any elements of the stream, the algorithm may toss some random coins.
\item The memory words read/written to (henceforth jointly referred to as probed) on any update operation $(i,\Delta)$ are completely determined from $i$ and the initially tossed random coins.
\end{itemize}
Thus a non-adaptive algorithm may not decide which memory words to probe based on the current contents of the memory words. Note that in this model of non-adaptivity, the random coins can encode e.g.\ a hash function chosen (independent of the stream) from a desirable family of hash functions and the update algorithm can choose what to probe based on the hash function. It is only the input-specific contents of the probed words that the algorithm may not use to decide which words to probe. To the best of our knowledge, all the known algorithms for classical turnstile streaming problems are indeed \emph{non-adaptive}, in particular \emph{linear sketches} (i.e.\ maintaining $\Pi v$ in memory for some $r\times n$ matrix $\Pi$, $r\ll n$). We further remark that our lower bounds only require the update procedure to be non-adaptive and still apply if adaptivity is used to answer queries.

For the remainder of the paper, for ease of presentation we assume that the update increments are (possibly negative) integers bounded by some $M \le \poly(n)$ in magnitude, and that the number of updates in the stream is also at most $\mathrm{poly}(n)$. We further assume the trans-dichotomous model \cite{FredmanW93,FredmanW94}, i.e.\ that the machine word size $w$ is $\Theta(\log n)$ bits. This is a natural assumption, since typically the streaming literature assumes that basic arithmetic operations on a value $|v_i|$, the index of the current position in the stream, or an index $i$ into $v$ can be performed in constant time.

We prove lower bounds for the following types of queries in turnstile streams. For each problem listed, the query function takes no input (other than point query, which takes an input $i\in[n]$). Each query below is accompanied by a description of what the data structure should output.

\begin{itemize}
\item {\bf $\ell_1$ heavy hitter:} return an index $i$ such that $|v_i| \ge \|v\|_{\infty} - \|v\|_1/2$.
\item {\bf Point query:} given $i$ at query time, returns $v_i \pm \|v\|_1 / 2$.
\item {\bf $\ell_p/\ell_q$ norm estimation ($1 \le q \le p \le \infty$):}  returns $\|v\|_p \pm \|v\|_q / 2$.
\item {\bf Entropy estimation.} returns a $2$-approximation of the entropy of the distribution which assigns probability $v_i/\|v\|_1$ to $i$ for each $i\in[n]$.
\end{itemize}

The lower bounds we prove for non-adaptive streaming algorithms are as follows ($n^{-O(1)} \le \delta \le 1/2-\Omega(1)$ is the failure probability of a query):
\begin{itemize}
\item Any randomized non-adaptive streaming algorithm for point query, $\ell_p/\ell_q$ estimation with $1\le q\le p\le \infty$, and entropy estimation, must have worst case update time $$t_u = \Omega\left(\frac{\log (1 /\delta)}{\sqrt{\lg n \log(eS/t_u)}}\right) .$$

We also show that any {\em deterministic} non-adaptive streaming algorithm for the same problems must have worst case update time $t_u = \Omega(\log n / \log(eS/t_u))$.

\item Any randomized non-adaptive streaming algorithm for $\ell_1$ heavy hitters, must have worst case update time $$t_u = \Omega\left(\min\left\{\sqrt{\frac{\log (1/\delta)}{\log(eS/t_u)}} , \frac{ \log (1/\delta)}{\sqrt{\lg t_u \cdot \lg(eS/t_u)}}\right\}\right).$$ Any deterministic non-adaptive streaming algorithm for the same problem must have worst case update time $t_u = \Omega\left(\frac{\log n}{\log(eS/t_u)}\right)$.
\end{itemize}
\begin{remark}
\textup{
The deterministic bound above for point query matches two previous upper bounds for point query~\cite{NNW14}, which use error-correcting codes to yield deterministic point query data structures. Specifically, for space $S = O(\log n)$, our lower bound implies $t_u=\Omega(\log n)$, matching an upper bound based on random codes. For space $O((\log n/\log \log n)^2)$, our lower bound implies $t_u = \Omega(\log n/\log \log n)$, matching an upper bound based on Reed-Solomon codes. Similarly, the deterministic bound above for deterministic $\ell_2/\ell_1$ norm estimation matches the previous upper bound for this problem~\cite{NNW14}, showing that for the optimal space $S = \Theta(\log n)$, the fastest query time of non-adaptive algorithms is $t_u=\Theta(\log n)$.
}

\textup{
These deterministic upper bounds are also in the cell probe model. In particular, the point query data structure based on random codes and the norm estimation data structure require access to combinatorial objects that are shown to exist via the probabilistic method, but for which we do not have explicit constructions. The point query structure based on Reed-Solomon codes can be implemented in the word RAM model with $t_u = \tilde{O}(\log n)$ using fast multipoint evaluation of polynomials. This is because performing an update, in addition to accessing $O(\log n/\log\log n)$ memory cells of the data structure, requires evaluating a degree-$O(\log n/\log\log n)$ polynomial on $O(\log n/\log\log n)$ points to determine which memory cells to access (see \cite{NNW14} for details).
}
\end{remark}

\begin{remark}
\textup{
The best known randomized upper bounds are $S = t_u = O(\log(1/\delta))$ for point query \cite{CM05} and $\ell_p/\ell_p$ estimation for $p\le 2$ \cite{ThorupZ12,KNPW11}. For entropy estimation the best upper bound has $S = t_u = \tilde{O}((\log n)^2\log(1/\delta))$ \cite{HNO08}. For $\ell_1$ heavy hitters the best known upper bound (in terms of $S$ and $t_u$) has $S = t_u = O(\log(n/\delta))$.
}
\end{remark}

In addition to being the first non-trivial turnstile update time lower bounds, we also managed to show that the update time has to increase polylogarithmically in $1/\delta$ as the error probability $\delta$ decreases, which is achieved with the typical reduction using $\lg(1/\delta)$ independent copies.

Our lower bounds can also be viewed in another light. If one is to obtain constant update time algorithms for the above problems, then one has to design algorithms that are \emph{adaptive}. Since all known upper bounds have non-adaptive updates, this would require a completely new strategy to designing turnstile streaming algorithms.

\paragraph{Technique.} 
As suggested by \Question{time-lb}, we prove our lower bounds using recent ideas in cell probe lower bound proofs. More specifically, we use ideas from the technique now formally known as \emph{cell sampling}~\cite{Gal03thecell, panigrahyMetric, Larsen12highercell}. This technique derives lower bounds based on one key observation: if a data structure/streaming algorithm probes $t$ memory words on an update, then there is a set $C$ of $t$ memory words such that at least $m/S^t$ updates probe only memory words in $C$, where $m$ is the number of distinct updates in the problem (for data structure lower bound proofs, we typically consider queries rather than updates, and we obtain tighter lower bounds by forcing $C$ to have near-linear size). 

We use this observation in combination with the standard one-way communication games typically used to prove streaming space lower bounds. In these games, Alice receives updates to a streaming problem and Bob receives a query. Alice runs her updates through a streaming algorithm for the corresponding streaming problem and sends the resulting $Sw$ bit memory footprint to Bob. Bob then answers his query using the memory footprint received from Alice. By proving communication lower bounds for \emph{any} communication protocol solving the one-way communication game, one obtains space lower bounds for the corresponding streaming problem.

At a high level, we use the cell sampling idea in combination with the one-way communication game as follows: if Alice's non-adaptive streaming algorithm happens to ``hash'' her updates such that they all probe the same $t$ memory cells, then she only needs to send Bob the contents of those $t$ cells. If $t < S$, this gives a communication saving over the standard reduction above. We formalize this as a general sketch-compression theorem, allowing us to compress the memory footprint of any non-adaptive streaming algorithm at the cost of increasing the error probability. This general theorem has the advantage of allowing us to re-use previous space lower bounds that have been proved using the standard reduction to one-way communication games, this time however obtaining lower bounds on the update time. We demonstrate these ideas in Section~\ref{sec:compress} and also give a more formal definition of the classic one-way communication game. 
\section{Sketch Compression}
\label{sec:compress}

In the following, we present a general theorem for compressing non-adaptive sketches. Consider a streaming problem in which we are to maintain an $n_0$-dimensional vector $v$. Let $U$ be the {\em update domain}, where each element of $U$ is a pair $(i,\Delta)\in [n] \times \{-M,\ldots, M\}$ for some $M=\poly(n)$. We interpret an update $(i,\Delta)\in U$ as having the effect $v[i] \leftarrow v[i] + \Delta$. Initially all entries of $v$ are $0$. We also define the {\em query domain} $Q = \{q_1,\ldots,q_r\}$, where each $q_i\in Q$ is a function $q_i:\Z^{n} \rightarrow\R$. With one-way communication games in mind, we define the input to a streaming problem as consisting of two parts. More specifically, the {\em pre-domain} $D_{pre}\subseteq U^a$ consists of sequences of $a$ update operations. The {\em post-domain} $D_{post}\subseteq \{U\cup Q\}^b$ consists of sequences of $b$ updates and/or queries. Finally, the {\em input domain} $D\subseteq D_{pre} \times D_{post}$ denotes the possible pairings of $a$ initial updates followed by $b$ intermixed queries and updates. The set $D$ defines a streaming problem $P_D$.

We say that a randomized streaming algorithm for a problem $P_D$ uses $S$ words of space if the maximum number of memory words used when processing any $d\in D$ is $S$. Here a memory word consists of $w = \Theta(\lg n)$ bits. The worst case update time $t_u$ is the maximum number of memory words read/written to upon processing an update operation for any $d\in D$. The error probability $\delta$ is defined as the maximum probability over all $d\in D$ and queries $q_i\in d\cap Q$, of returning an incorrect results on query $q_i$ after the updates preceding it in the sequence $d$.

A streaming problem $P_D$ of the above form naturally defines a one-way communication game: on an input $(d_1,d_2)\in D$, Alice receives $d_1$ (the first $a$ updates) and Bob receives $d_2$ (the last $b$ updates and/or queries). Alice may now send a message to Bob based on her input and Bob must answer all queries in his input as if streaming through the concatenated sequence of operations $d_1\circ d_2$. The error probability of a communication protocol is defined as the maximum over all $d\in D$ and $q_i\in \{d\cap Q\}$, of returning an incorrect results on $q_i$ when receiving $d$ as input.

Traditionally, the following reduction is used:

\begin{theorem}
If there is a randomized streaming algorithm for $P_D$ with space usage $S$ and error probability $\delta$, then there is a public coin protocol for the corresponding one-way communication game in which Alice sends $Sw$ bits to Bob and the error probability is $\delta$.
\end{theorem}
\begin{proof}
Alice simply runs the streaming algorithm on her input and sends the memory image to Bob. Bob continues the streaming algorithm and outputs the answers.
\end{proof}

Recall from Section~\ref{sec:intro} that a randomized streaming algorithm is {\em non-adaptive} if:
\begin{itemize}
\item Before processing any elements of the stream, the algorithm may toss some random coins.
\item The memory words read/written to (henceforth jointly referred to as probed) on any update operation $(i,\Delta)$ is completely determined from $i$ and the initially tossed random coins.
\end{itemize}
We show that for non-adaptive algorithms, one can efficiently reduce the communication by increasing the error probability. We require some additional properties of the problem however: we say that a streaming problem $P_D$ is {\em permutation invariant} if:
\begin{itemize}
\item For any {\em permutation} $\pi:[n]\rightarrow[n]$, it holds that $\pi(q_i(v))=(\pi(q_i)(\pi(v)))$ for all $q_i\in Q$. Here $\pi(v)$ is the $n$-dimensional vector with value $v[i]$ in entry $\pi[i]$, $\pi(q_i)$ maps all indices (if any) in the definition of the query $q_i$ wrt. $\pi$ and $\pi(q_i(v))$ maps all indices in the answer $q_i(v)$ (if any) wrt. $\pi$.
\end{itemize}
Observe that point query, $\ell_p$ estimation, entropy estimation and heavy hitters all are permutation invariant problems. For point query, we have $\pi(q_i(v))=q_i(v)$ since answers contain no indices, but $\pi(q_i)$ might differ from $q_i$ since queries are defined from indices. For $\ell_p$ estimation and entropy estimation, we simply have $\pi(q_i(v))=q_i(v)$ and $\pi(q_i)=q_i$ since neither queries or answers involve indices. For heavy hitters we have $\pi(q_i)=q_i$ (there is only one query), but we might have that $\pi(q_i(v)) \neq q_i(v)$ since the answer to the one query is an index. We now have the following:

\begin{theorem}
\label{thm:main}
If there is a randomized non-adaptive streaming algorithm for a permutation invariant problem $P_D$ with $a \leq \sqrt{n}$, having space usage $S$, error probability $\delta$, and worst case update time $t_u\leq (1/2)(\lg n/\lg(eS/t_u)) $, then there is a private coin protocol for the corresponding one-way communication game in which Alice sends at most $a\lg e + t_ua\lg(eS/t_u) + \lg a + \lg\lg (en/a) + 2t_u w + 1$ bits to Bob and the error probability is $2e^{a}\cdot (eS/t_u)^{t_ua}\delta$.
\end{theorem}

Before giving the proof, we present the two main ideas. First observe that once the random choices of a non-adaptive streaming algorithm have been made, there must be a large collection of indices $I \subseteq [n]$ for which all updates $(i,\Delta)$, where $i \in I$, probe the same small set of memory words (there are at most $\binom{S}{t_u}$ distinct sets of $t_u$ words to probe). If all of Alice's updates probed only the same set of $t_u$ words, then Alice could simply send those words to Bob and we would have reduced the communication to $t_uw$ bits. To handle the case where Alice's updates probe different sets of words, we make use of the permutation invariant property. More specifically, we show that Alice and Bob can agree on a collection of $k$ permutations of the input indices, such that one of these permutes all of Alice's updates to a new set of indices that probe the same $t_u$ memory cells. Alice can then send this permutation to Bob and they can both alter their input based on the permutation. Therefore Alice and Bob can solve the communication game with $\lg k + t_uw$ bits of communication. The permutation of indices unfortunately increases the error probability as we shall see below. With these ideas in mind, we now give the proof of Theorem~\ref{thm:main}.
\\\\
\begin{proofof}{Theorem~\ref{thm:main}}
Alice and Bob will permute their communication problem and use the randomized non-adaptive streaming algorithm on this transformed instance to obtain an efficient protocol for the original problem.

By Yao's principle, we have that the randomized complexity of a private coin one-way communication game with error probability $\delta$ equals the complexity of the best deterministic algorithm with error probability $\delta$ over the worst distribution. Hence we show that for any distribution $\mu$ on $D\subseteq D_{pre}\times D_{post}$, there exists a deterministic one-way communication protocol with $t_ua\lg S + \lg a + \lg\lg n + t_u w$ bits of communication and error probability $2e^{a}\cdot (eS/t_u)^{t_ua}\delta$. We let $\mu_1$ denote the marginal distribution over $D_{pre}$ and $\mu_2$ the marginal distribution over $D_{post}$.

Let $\mu = (\mu_1,\mu_2)$ be a distribution on $D$. Define from this a new distribution $\gamma = (\gamma_1,\gamma_2)$: pick a uniform random permutation $\Perm$ of $[n]$. Now draw an input $d$ from $\mu$ and permutate all indices of updates (and queries if defined for such) using the permutation $\Perm$. The resulting sequence $\Perm(d)$ is given to Alice and Bob as before, which defines the new distribution $\gamma = (\gamma_1,\gamma_2)$. We use $A\sim \mu_1$ to denote the r.v.\ providing Alice's input drawn from distribution $\mu_1$ and $\Perm(A) \sim \gamma_1$ denotes the random variable providing Alice's transformed input. We define $B\sim \mu_2$ and $\Perm(B)\sim \gamma_2$ symmetrically.

Recall we want to solve the one-way communication game on $A$ and $B$. To do this, first observe that by fixing the random coins, the randomized non-adaptive streaming algorithm gives a non-adaptive and deterministic streaming algorithm that has error probability $\delta$ and space usage $S$ under distribution $\gamma$. Before starting the communication protocol on $A$ and $B$, Alice and Bob both examine the algorithm (it is known to both of them). Since it is non-adaptive and deterministic, they can find a set of $t_u$ memory words $C$, such that at least
$$
\frac n{\binom{S}{t_u}} > \frac n{(eS/t_u)^{t_u}} \ge \sqrt{n} \ge a
$$
indices $i\in[n]$ satisfy that any update $(i,\Delta)$ probes only memory words in $C$. We let $I^C$ denote the set of all such indices (again, $I^C$ is known to both Alice and Bob). Alice and Bob also agree on a set of permutations $\{\rho_1,\ldots,\rho_k\}$ (we determine a value for $k$ later), such that for any set of at most $a$ indices, $I'$, that can occur in Alice's updates, there is at least one permutation $\rho_i$ where:
\begin{itemize}
\item $\rho_i(j) \in I^C$ for all $j\in I'$
\item Let $I(A)$ denote the (random) indices of the updates in $A$. Then the probability that the non-adaptive and deterministic protocol errs on input $\rho_i(A)$, conditioned on $I(A) = I'$, is at most $2e^{a}\cdot (eS/t_u)^{t_ua}\cdot \eps_{I(A) = I'}$ where $\eps_{I(A)=I'}$ is the error probability of the deterministic and non-adaptive streaming algorithm on distribution $\gamma$, conditioned on $I(A) = I'$.
\end{itemize}
Again, this set of permutations is known to both players. The protocol is now simple: upon receiving $A$, Alice finds the index $i$ of a permutation $\rho_i$ satisfying the above for the indices $I(A)$. She then sends this index to Bob and runs the deterministic and non-adaptive algorithm on $\rho_i(A)$. She forwards the addresses and contents of all memory words in $C$ as well. This costs a total of $\lg k + |C|(w+\lg S(n)) \le \lg k +2t_u(n) w$ bits. Note that no words outside $C$ are updated during Alice's updates. Bob now remaps his input $B$ according to $\rho_i$ and runs the deterministic and non-adaptive streaming algorithm on his updates and queries. Observe that for each query $q_j \in B$, Bob will get the answer $\rho_i(q_j)(\rho_i(v))$ if the algorithm does not err, where $v$ is the ``non-permuted'' vector after processing all of Alice's updates $A$ and all updates in $B$ preceeding the query $q_j$. For each such answer, he computes $\rho_i^{-1}(\rho_i(q_j)(\rho_i(v)))$. Since $P_D$ is permutation invariant, we have $\rho_i^{-1}(\rho_i(q_j)(\rho_i(v)))=\rho_i^{-1}(\rho_i(q_j(v)))=q_j(v)$. The final error probability (over $\mu$) is hence at most $2e^{a}\cdot (eS/t_u)^{t_ua}\cdot \delta$ since $\E_{I'} \eps_{I(A)=I'} = \delta$.

We only need a bound on $k$. For this, fix one set $I'$ of at most $a$ indices in $[n]$ and consider drawing $k$ uniform random permutations. For each such random permutation $\Gamma$, note that $\Gamma(I')$ is distributed as $I(\Perm(A))$ conditioned on $I(A) = I'$. Hence, the expected error probability when using the map $\Gamma$ (expectation over choice of $\Gamma$) is precisely $\eps_{I(A) = I'}$. We also have
$$\Pr(\Gamma(I')\subseteq I^C) = \frac{\binom{|I^C|}{|I'|}}{\binom{n}{|I'|}} \ge \left(\frac{|I^C|}{en}\right)^{|I'|} \ge e^{-a}\cdot \left(\frac{eS}{t_u}\right)^{-t_ua}$$
By Markov's inequality and a union bound, we have both $\Gamma(I')\subseteq I^C$ and error probability at most $2e^{a}\cdot (eS/t_u)^{t_ua}\cdot\eps_{I(A)=I'}$ with probability at least $e^{-a}\cdot (eS/t_u)^{-t_ua}/2 \eqdef p$ over the choice of $\Gamma$. Thus if we pick $k = (1/p)a\log(en/a) > (1/p)\cdot \log\binom{n}{a}$ and use that $1+x\le e^x$ for all real $x$, then setting the probability that all permutations $\Gamma$ chosen fail to have the desired failure probability and $\Gamma(I')\subseteq I^C$ is at most $(1-p)^k < 1/\binom{n}{a}$. Thus by a union bound over all $\binom{n}{a}$ size-$a$ subsets of indices, we can conclude that the desired set of
permutations exists.
\end{proofof}

\section{Implications}
\label{sec:implications}

In this section, we present the (almost) immediate implications of Theorem~\ref{thm:main}. All these result follow by re-using the one-way communication game lower bounds originally proved to obtain space lower bounds of heavy hitters. Consider the generic protocol in Figure~\ref{fig:decode}. For different streaming problems, we will show the different ways to implement the function $Check(t,j)$ using the updates and queries of the problems, where $Check(t,j)$ is supposed to be able to tell if $t$ is equal to $i_j$ with failure probability $\delta$. We show first that if this is the case, we obtain a lower bound on the update time $t_u$. The lower bounds then follow from the implementation of $Check(t,j)$.

\begin{figure}[ht]
        \begin{algorithmic}[1]
	\State Alice chooses $a$ indices $i_1,\ldots,i_a$ randomly without replacement in $[n]$.
	\State Alice performs updates $v[i_j] \leftarrow v[i_j]+C^j$ for $j=1,\ldots,a$, where $C$ is a large constant.
	\State Alice sends her memory state to Bob.
	\For{$j$ from $a$ down to 1}
		\Comment Bob decodes $i_a, \ldots, i_1$ from Alice's message.	
		\ForAll{$t\in [n]$}
			\If{Check(t,j)}
				\State Bob declares $i_j = t$.
				\State Bob performs the update $v[i_j] \gets v[i_j] - C^j$	 
			\EndIf
		\EndFor
           \EndFor
        \end{algorithmic}
\caption{\label{fig:decode}The communication protocol for Alice to send $a$ random numbers in $[n]$ to Bob.
}
\end{figure}

\begin{theorem}\label{thm:randomized-bound}
Any randomized non-adaptive streaming algorithm that can implement all the $Check(t,j)$ calls using no more than $k \leq n^{O(1)}$ queries with failure probability $n^{-\Omega(1)} \le \delta \leq 1/2-\Omega(1)$ each, must have worst case update time $t_u =\Omega\left(\min\left\{\sqrt{\frac{\lg 1/\delta}{\lg(eS/t_u)}}, \frac{\lg 1/\delta}{\sqrt{\lg k \lg(eS/t_u)}}\right\}\right)$.
\end{theorem}
\begin{proof}
We first prove the lower bound for the case where $\sqrt{\frac{\lg 1/\delta}{\lg(eS/t_u)}}$ is the smaller of the two terms. This is the case at least from $n^{-\Omega(1)} \leq \delta \leq k^{-3}$. Assume for contradiction that such a randomized non-adaptive algorithm exists with $t_u=o(\sqrt{\lg(1/\delta)/ \lg(eS/t_u)})$. Set $a=c_0 t_u$ for a sufficiently large constant $c_0$. We invoke Theorem~\ref{thm:main} to conclude that Alice can send her memory state to Bob using $a\lg e + t_ua\lg(eS/t_u) + \lg a + \lg\lg (en/a) + 2t_u w + 1$ bits while increasing the failure probability of Bob's $k$ queries to
$$2e^{a}\cdot (eS/t_u)^{t_u a} \delta \leq (eS/t_u)^{c_0 t_u^2} \delta^{1-o(1)} \leq \delta^{1-o(1)}\leq k^{-2}$$
each. By a union bound, the answer to all Bob's queries are correct with probability at least $1-1/k \geq 9/10$, so Bob can recover $i_1,\ldots,i_a$ with probability $9/10$. By Fano's inequality, Alice must send Bob at least $\Omega(H(i_1,\ldots, i_a)) = \Omega(a\log(n/a)) \geq c_0c_1 t_u \lg n$ bits for some constant $c_1$ (where $c_1$ does not depend on $c_0$). But the size of her message was
$$
a\lg e + t_ua\lg(eS/t_u) + \lg a + \lg\lg (en/a) + 2t_u w + 1 \le c_0 t_u^2 \lg(eS/t_u) + 2t_uw + o(t_u \lg n) .
$$
We assumed $t_u = o\left(\sqrt{\lg(1/\delta)/\lg(eS/t_u)}\right) = o\left(\sqrt{\lg n/\lg(eS/t_u)}\right)$ and thus the above is bounded by $2t_uw + o(t_u \lg n)$. Setting $c_0$ high enough, we get that $2w \ge (c_0c_1\lg n)/2$ and thus we have reached a contradiction.

For the case $k^{-3} < \delta < 1/2-\Omega(1)$, observe that we can decrease the failure probability to $k^{-3}$ by increasing the space, update time and query time by a factor $\alpha=\Theta(\lg_{1/\delta} k)$: simply run $\alpha$ independent copies of the streaming algorithm and return the majority answer on a query. Hence the lower bound becomes 
$$\Omega\left(\frac{\sqrt{\frac{\lg k}{\lg(eS\alpha/(t_u \alpha))}}}{\alpha}\right) = \Omega\left(\frac{\lg 1/\delta}{\sqrt{\lg k \lg(eS/t_u)}}\right).$$
\end{proof}

\begin{theorem}\label{thm:det-bound}
Any deterministic non-adaptive streaming algorithm that can implement $Check(t,j)$ must have worst case update time 
$$t_u =\Omega\left(\frac{\log n}{\log(eS/t_u)}\right) .$$
\end{theorem}

\begin{proof}
The proof is similar to that of Theorem~\ref{thm:randomized-bound}. Assume for contradiction there is a deterministic non-adaptive streaming algorithm with $t_u = o(\lg n/\lg(eS/t_u))$. Choose $a=c_0 t_u$ for sufficiently large constant $c_0$. By Theorem~\ref{thm:main}, Alice can send her memory state to Bob using 
\begin{align*}
a\lg e + t_ua\lg(eS/t_u) + \lg a + \lg\lg (en/a) + 2t_u w + 1 &\le c_0 t_u^2 \lg(eS/t_u) + 2t_uw + o(t_u \lg n)\\
{} &\le 2t_u w + o(t_u \lg n)
\end{align*}
bits and Bob can still answer every query correctly. Since Bob can recover $i_1, \ldots, i_a$, Alice must send Bob at least $H(i_1,\ldots, i_a) = \Omega(a\log(n/a)) = c_0c_1 t_u \log n$ for some constant $c_1$ (independent of $c_0$). Setting $c_0$ large enough gives $2w \ge (c_0 c_1 \lg n)/2$, thus deriving the contradiction.
\end{proof}

\begin{corollary}
Any randomized non-adaptive streaming algorithm for point query, $\ell_p$ estimation with $1\le p\le 2$, and entropy estimation with failure probability $n^{-\Omega(1)} \le \delta \le 1/2-\Omega(1)$ must have worst case update time $$t_u = \Omega\left(\frac{\log (1 /\delta)}{\sqrt{\lg n \log(eS/t_u)}}\right) .$$
Any deterministic non-adaptive streaming algorithm for the same problems must have worst case update time $t_u = \Omega(\log n / \log(eS/t_u))$.
\end{corollary}
\begin{proof}
We just need to show implementations of $Check(t, j)$ that use $n^{O(1)}$ queries in total and apply Theorem~\ref{thm:randomized-bound} to get the stated bound for randomized algorithms. The stated bound for deterministic algorithms follows from applying Theorem~\ref{thm:det-bound} instead.

\subsection{Applications to specific problems}

\paragraph{Point query.} We can implement each $Check(t,j)$ by simply querying for the value of $v[t]$ and check if the returned value is at least $C^j / 3$. By the guarantee of the algorithm, if it does not fail, the returned value is within a factor $3$ from the right value and thus, $Check(t,j)$ can correctly tell if $t=v_j$. Thus we run a total of $n a = n^{O(1)}$ queries to implement all $Check(t,j)$'s.

\paragraph{$\ell_p/\ell_q$ estimation.} We can implement $Check(t,j)$ as follows.
\begin{itemize}
\item $v[t] \gets v[t] - C^j$
\item Check if $\|v\|_p$ is at most $C^j / 3$
\item $v[t] \gets v[t] + C^j$
\end{itemize}
By the guarantee of the algorithm, if it does not fail, the estimated $\ell_p / \ell_q$ norm is at most $3C^{j-1} < C^j / 3$ if $t=i_j$ and it is greater than $C^j/3$ if $t\ne i_j$. Thus, $Check(t,j)$ can correctly tell if $t=v_j$. Again we used $n^{O(1)}$ queries in total.

\paragraph{Entropy estimation.} We can implement $Check(t,j)$ as follows.
\begin{itemize}
\item $v[t] \gets v[t] - C^j$
\item Check if the entropy is at most 1/3
\item $v[t] \gets v[t] + C^j$
\end{itemize}
Consider the case the algorithm does not fail. First, if $t=i_j$ then the entropy is at most
\begin{align*}
\sum_{i=1}^j \frac{C^i (C-1)}{C^{j+1}-C}\log \frac{C^{j+1}-C}{C^i (C-1)} &\le \sum_{i=1}^{j-1} \frac{C^i (C-1)}{C^{j+1}-C}\log 2C^{j-i} + \frac{C^j (C-1)}{C^{j+1}-C}\log \frac{C^{j+1}-C}{C^j (C-1)}\\
&\le O(\frac{\log C}{C-1}) + \frac{C^j (C-1)}{C^{j+1}-C} \cdot \frac{C^j-C}{C^j(C-1)} \le O(\frac{\log C}{C-1}) < 1/10
\end{align*}
On the other hand, if $t\ne i_j$ then after the first operation, $\|v\|_1 \le \frac{2C^{j+1}-C^j-C}{C-1} < 3 C^j$ so the entropy is at least the binary entropy of whether the index $i_j$ is picked, which is greater than $H(1/3) > 0.9$. Thus, by the guarantee of the algorithm, $Check(t,j)$ correctly tells if $t=v_j$ using $n^{O(1)}$ queries.
\end{proof}

\begin{corollary}
Any randomized non-adaptive streaming algorithm for $\ell_1$ heavy hitters with failure probability $n^{-O(1)} \le \delta \le 1/2-\Omega(1)$ must have worst case update time $$t_u = \Omega\left(\min\left\{\sqrt{\frac{\log (1/\delta)}{\log(eS/t_u)}} , \frac{ \log (1/\delta)}{\sqrt{\lg t_u \cdot \lg(eS/t_u)}}\right\}\right).$$ Any deterministic non-adaptive streaming algorithm for the same problems must have worst case update time $t_u = \Omega\left(\frac{\log n}{\log(eS/t_u)}\right)$.
\end{corollary}
\begin{proof}
We use a slightly different decoding procedure for Bob. Instead of running $Check(t,j)$ all indices $t\in[n]$, in iteration $j$, Bob can simply query for the heavy hitter to find $i_j$. Note that in iteration $j$, we have $\|v\|_1 = \frac{C^{j+1}-C}{C-1} < 2C^j = 2v[i_j]$ so if the algorithm is correct, Bob will find the index $i_j$. We have thus implemented all the $Check(t,j)$'s using only $a$ queries. Recall from the proof of Theorem~\ref{thm:randomized-bound} that $a=O(t_u)$.
\end{proof}

\bibliographystyle{alpha}

\bibliography{biblio}

\end{document}